\documentclass[runningheads,envcountsame,a4paper]{llncs}
\usepackage[T1]{fontenc}
\usepackage[latin9]{inputenc}
\usepackage{amsmath}
\usepackage{amssymb}
\usepackage{wasysym}
\usepackage{graphicx}

\makeatletter

\providecommand{\tabularnewline}{\\}

\usepackage{amsmath}
\usepackage{epstopdf}

\DeclareMathOperator{\cra}{\mathrm{-cl-RA}}
\DeclareMathOperator{\Cra}{\mathrm{cl-RA}}

\makeatother

\begin{document}
\title{Two Results on Discontinuous Input Processing}
\titlerunning{Complexity of Road Coloring with Prescribed Reset Words}

\author{Vojt\v{e}ch Vorel\inst{1}\thanks{Research supported by the Czech Science Foundation grant GA14-10799S and the GAUK grant No. 52215.}}

\authorrunning{Vojt\v{e}ch Vorel}
\tocauthor{Vojt\v{e}ch Vorel}

\institute{Faculty of Mathematics and Physics, Charles University, Malostransk\'{e} n\'{a}m. 25, Prague, Czech Republic,\\ \email{vorel@ktiml.mff.cuni.cz}}

\maketitle  
\begin{abstract}
First, we show that universality and other properties of general jumping finite automata are undecidable, which answers a question asked by Meduna and Zemek in 2012.
Second, we close the study raised by \v{C}erno and Mr\'{a}z in 2010 by proving that clearing restarting automata using contexts of size two can accept binary non-context-free languages.
\end{abstract}
\renewcommand{\arraystretch}{1.6}

\section{Introduction}

In 2012, Meduna and Zemek \cite{athMED1,athMED1book} introduced \emph{general
jumping finite automata} as a model of discontinuous information processing.
A general jumping finite automaton (GJFA) is described by a finite
set $Q$ of states, a finite alphabet $\Sigma$, a finite set $R$
of \emph{rules }from $Q\times\Sigma^{\star}\times Q$, an initial
state $q_{0}\in Q$, and a set $F\subseteq Q$ of final states. In
a step of computation, the automaton switches from a state $r$ to
a state $s$ using a rule $\left(r,v,s\right)\in R$, and deletes
a factor equal to $v$ from any part of the input word. The choices
of the rule used and of the factor deleted are made nondeterministically
(in other words, the read head can jump to any position). A word is
accepted if there is a computation resulting in the empty word. The
boldface term $\mathbf{GJFA}$ refers to the class of languages accepted
by GJFA. The initial work \cite{athMED1,athMED1book} deals mainly
with closure properties of \textbf{$\mathbf{GJFA}$} and its relations
to classical language classes (the publications contain flaws, see
\cite{VO6}). It turns out that the class $\mathbf{GJFA}$ does not
have Boolean closure properties (complementation, intersection) nor
closure properties related to continuous processing (concatenation,
Kleene star, homomorphism, inverse homomorphism, shuffle). Accordingly,
the class also does not stick to classical complexity measures - it
is incomparable with both regular and context-free languages. It is
a proper subclass of both context-sensitive languages and of the class
NP, while there exist NP-complete \textbf{$\mathbf{GJFA}$ }languages.
See \cite{athFER1}, which is an extended version of \cite{athFER1conf}.

On the other hand, the concept of \emph{restarting automata }\cite{reaJAN2,reaJAN1}
is motivated by\emph{ }reduction analysis and grammar checking of
natural language sentences. In 2010, \v{C}erno and Mr\'{a}z \cite{reaCER2}
introduced a subclass named \emph{clearing restarting automata }($\Cra$)
in order to describe systems that use only very basic types of reduction
rules (see also \cite{reaCER3}). Unlike GJFA, clearing restarting
automata may delete factors according to contexts and endmarks, but
they are not controlled by state transitions. A key property of each
$\Cra$ is the maximum length of context used. For $k\ge1$, a \emph{$k$-clearing
restarting automaton} ($k\cra$) is described by a finite alphabet
$\Sigma$ and a finite set $I$ of instructions of the form $\left(u_{\mathrm{L}},v,u_{\mathrm{R}}\right)$,
where $v\in\Sigma^{*}$, $u_{\mathrm{L}}\in\Sigma^{k}\cup\cent\Sigma^{k-1}$,
and $u_{\mathrm{R}}\in\Sigma^{k}\cup\Sigma^{k-1}\$$. The words $u_{\mathrm{L}},u_{\mathrm{R}}$
specify left and right context for consuming a factor $v$, while
$\cent$ and $\$$ stand for the left and right end of input, respectively.

\section{Preliminaries}

We heavily use the notion of \emph{insertion}, as it was described,
e.g., in \cite{reaEHR2,reaHAU1,athKAR2}:
\begin{definition}
Let $K,L\subseteq\Sigma^{\star}$ be languages. The \emph{insertion
}of $K$ to $L$ is 
\[
L\leftarrow K=\left\{ u_{1}vu_{2}\mid u_{1}u_{2}\in L,v\in K\right\} .
\]
More generally, for each $k\ge1$ we denote
\begin{eqnarray*}
L\leftarrow^{k}K & = & \left(L\leftarrow^{k-1}K\right)\leftarrow K,\\
L\leftarrow^{\star}K & = & \bigcup_{i\ge0}L\leftarrow^{i}K,
\end{eqnarray*}
where $L\leftarrow^{0}K$ stands for $L$. In expressions with $\leftarrow$
and $\leftarrow^{\star}$, a singleton set $\left\{ w\right\} $ may
be replaced by $w$. 

A chain $L_{1}\leftarrow L_{2}\leftarrow\dots\leftarrow L_{d}$ of
insertions is evaluated from the left, e.g. $L_{1}\leftarrow L_{2}\leftarrow L_{3}$
means $\left(L_{1}\leftarrow L_{2}\right)\leftarrow L_{3}$. Finally,
$L\subseteq\Sigma^{\star}$ is a\emph{ unitary language} if $L=w\leftarrow^{\star}K$
for $w\in\Sigma^{\star}$ and finite $K\subseteq\Sigma$.
\end{definition}
As described above, a GJFA is a quintuple $M=\left(Q,\Sigma,R,q_{0},F\right)$.
The original definition of the accepted language $L\!\left(M\right)$
is based on \emph{configurations }that specify a position of the read
head (i.e., starting positions of the factor to be erased in the next
step). For our proofs, this type of configurations is useless, whence
we save space by directly using the following generative characterization
\cite[Corollary 1]{VO6} of $L\!\left(M\right)$ as a definition:
\begin{definition}
\label{lem: paths vs accept}Let $M=\left(Q,\Sigma,R,s,F\right)$
be a GJFA and $w\in\Sigma^{*}$. Then $w\in\mathrm{L}\!\left(M\right)$
if and only if $w=\epsilon$ and $s\in F$, or 
\begin{equation}
w\in\epsilon\leftarrow v_{d}\leftarrow v_{d-1}\leftarrow\cdots\leftarrow v_{2}\leftarrow v_{1},\label{eq:comp}
\end{equation}
where $d\ge1$ and $v_{1},v_{2},\dots,v_{d}$ is a labeling of an
accepting path in $M$.
\end{definition}

\begin{definition}
For $k\ge0$, a \emph{$k$-context rewriting system }is a tuple $R=\left(\Sigma,\Gamma,I\right)$,
where $\Sigma$ is an input alphabet, $\Gamma\supseteq\Sigma$ is
a working alphabet not containing the special symbols $\cent$ and
$\$$, called \emph{sentinels}, and $I$ is a fi{}nite set of instructions
of the form
\[
\left(u_{\mathrm{L}},v\rightarrow t,u_{\mathrm{R}}\right),
\]
where where $u_{\mathrm{L}}$ is a \emph{left context}, $x\in\Gamma^{k}\cup\cent\Gamma^{k-1}$,
$y$ is a \emph{right context}, $y\in\Gamma^{k}\cup\Gamma^{k-1}\$$,
and $v\rightarrow t$ is a \emph{rule}, $z,t\in\Gamma^{\star}$. A
word $w=u_{1}vu_{2}$ can be rewritten into $u_{1}tu_{2}$ (denoted
as$u_{1}vu_{2}\rightarrow_{R}u_{1}tu_{2}$) if and only if there exists
an instruction $\left(u_{\mathrm{L}},v\rightarrow t,u_{\mathrm{R}}\right)\in I$
such that $u_{\mathrm{L}}$ is a suffi{}x of $\cent u_{1}$ and $u_{\mathrm{R}}$
is a prefi{}x of $u_{2}\$$. The symbol $\rightarrow_{R}^{\star}$
denotes the reflexive-transitive closure of $\rightarrow_{R}$.
\end{definition}

\begin{definition}
For $k\ge0$, a $k$\emph{-clearing restarting automaton} ($k\cra$)
is a system $M=(\Sigma,I)$, where $(\Sigma,\Sigma,I)$ is a $k$-context
rewriting system such that for each $\mathbf{i}=\left(u_{\mathrm{L}},v\rightarrow t,u_{\mathrm{R}}\right)\in I$
it holds that $v\in\Sigma^{+}$ and $t=\epsilon$. Since $t$ is always
the empty word, we use the notation $\mathbf{i}=\left(u_{\mathrm{L}},v,u_{\mathrm{R}}\right)$.
A $k\cra$ $M$ accepts the language 
\[
L\!\left(M\right)=\left\{ w\in\Sigma^{\star}\mid w\vdash_{M}^{\star}\epsilon\right\} ,
\]
where $\vdash_{M}$ denotes the rewriting relation $\rightarrow_{\overline{M}}$
of $\overline{M}=(\Sigma,\Sigma,I)$. The term $\mathcal{L}\!\left(k\cra\right)$
denotes the class of languages accepted by $k\cra$.
\end{definition}
Like in GJFA, one may consider the generative approach to languages
accepted by clearing restarting automata. In this case, the generative
approach is formalized by writing $w_{2}\dashv w_{1}$ instead of
$w_{1}\vdash w_{2}$.

\section{Undecidability in General Jumping Finite Automata\label{sec:Decidability-in-General}}

This section proves the following theorem, which solves an open problem
stated in \cite{athMED1book,athMED1}:
\begin{theorem}
Given a GJFA $M=\left(Q,\Sigma,R,s,F\right)$, it is undecidable whether
$L\!\left(M\right)=\Sigma^{*}$.\end{theorem}
\begin{proof}
Given a context-free grammar $G$ with terminal alphabet $\Sigma_{\mathrm{T}}$,
it is undecidable whether $\mathrm{L}\!\left(G\right)=\Sigma_{\mathrm{T}}^{*}$
\cite{HU1}. We present a reduction from this problem to the universality
of GJFA. Assume that the given grammar $G$:
\begin{itemize}
\item has non-terminal alphabet $\Sigma_{\mathrm{N}}=\left\{ A_{1},\dots,A_{m}\right\} $
with a start symbol $A_{\mathrm{S}}\in\Sigma_{\mathrm{N}}$,
\item does not accept the empty word, and
\item is given in Greibach normal form \cite{HU1} as
\[
B_{i}\rightarrow u_{i},
\]
where $B_{i}\in\Sigma_{\mathrm{N}}$ and $u_{i}\in\Sigma_{\mathrm{T}}\Sigma_{\mathrm{N}}^{*}$
for $i\in\left\{ 1,\dots,n\right\} $.
\end{itemize}

We construct a GJFA $M_{G}=\left(Q,\Gamma,R,s,F\right)$ as follows,
denoting $\Sigma_{\mathrm{B}}=\left\{ b_{1},\dots,b_{m}\right\} $:
\begin{eqnarray*}
Q & = & \left\{ q_{0},q_{1},q_{2},q_{3},q_{4}\right\} ,\\
\Gamma & = & \Sigma_{\mathrm{T}}\cup\Sigma_{\mathrm{N}}\cup\Sigma_{\mathrm{B}},
\end{eqnarray*}
$s=q_{0}$, $F=\left\{ q_{4}\right\} $, and $R$ follows Figure \ref{fig:MG}.
\begin{figure}
\begin{centering}
\includegraphics{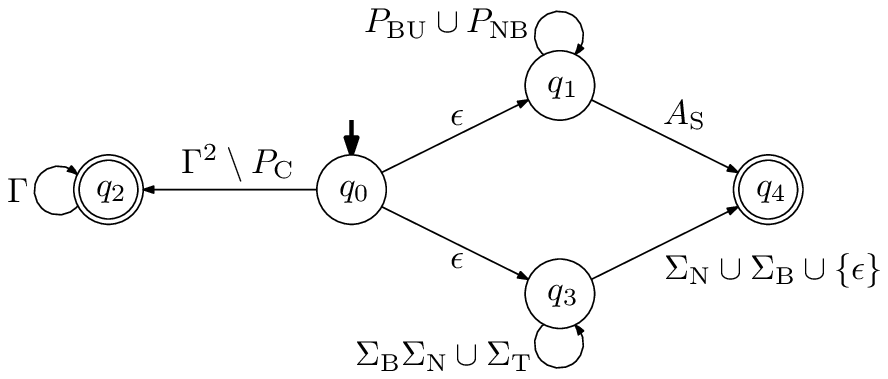}
\par\end{centering}

\caption{The GJFA $M_{G}$\label{fig:MG}}
\end{figure}
 Each arrow labeled with a finite set $S\subseteq\Gamma^{*}$ stands
for $\left|S\right|$ transitions, each labeled by a unique $v\in S$.
The following finite sets are used:
\begin{eqnarray*}
P_{\mathrm{BU}} & = & \left\{ b_{i}u_{i}\mid i=1,\dots,m\right\} ,\\
P_{\mathrm{NB}} & = & \left\{ A_{i}b_{i}\mid i=1,\dots,m\right\} ,\\
P_{\mathrm{C}} & = & \left\{ xA_{1}\mid x\in\Sigma_{\mathrm{T}}\right\} \cup\\
 & \cup & \left\{ A_{i}b_{i}\mid i=1,\dots,m\right\} \cup\\
 & \cup & \left\{ b_{i}A_{i+1}\mid i=1,\dots,m-1\right\} \cup\\
 & \cup & \left\{ b_{m}x\mid x\in\Sigma_{\mathrm{T}}\right\} .
\end{eqnarray*}
For a word $w\in\Gamma^{*}$ we denote with $w_{\mathrm{T}}$ and
$w_{\mathrm{N,B}}$ the projections of $w$ to subalphabets $\Sigma_{\mathrm{T}}$
and $\Sigma_{\mathbf{N}}\cup\Sigma_{\mathbf{B}}$ respectively.%
\footnote{A \emph{projection }to $\Sigma\subseteq\Gamma$ is given by the homomorphism
that maps $x\in\Gamma$ to $x$ if $x\in\Sigma$ or to $\epsilon$
otherwise.%
} Let us show that $L\!\left(G\right)=\Sigma_{\mathrm{T}}^{*}$ if
and only if $L\!\left(M_{G}\right)=\Gamma^{*}$.

First, suppose that $L\!\left(G\right)=\Sigma_{\mathrm{T}}^{*}$ and
take an arbitrary $w\in\Gamma^{*}$. Describe a derivation of $w_{\mathrm{T}}$
by $G$ using $v_{0},v_{1},\dots,v_{d}\in\Sigma_{\Gamma}$, $d\ge1$,
where 
\begin{eqnarray*}
v_{0} & = & A_{\mathrm{S}},\\
v_{d} & = & w_{\mathrm{T}},\\
v_{k} & = & v_{\mathbf{p},k}A_{i_{k}}v_{\mathbf{s},k},\\
v_{k+1} & = & v_{\mathbf{p},k}u_{i_{k}}v_{\mathbf{s},k}
\end{eqnarray*}
for each $k\in\left\{ 0,\dots,d-1\right\} $. For $k\in\left\{ 0,\dots,d\right\} $,
we define inductively a word $w_{k}\in\Sigma_{\Gamma}$ satisfying
$\left(w_{k}\right)_{\mathrm{T}}=v_{k}$ as follows. First, $w_{0}=A_{\mathrm{S}}$.
Next, take $0\le k\le d-1$ and write $w_{k}=w_{\mathbf{p},k}A_{i_{k}}w_{\mathbf{s},k}$
such that $\left(w_{\mathbf{p},k}\right)_{\mathrm{T}}=v_{\mathbf{p},k}$
and $\left(w_{\mathbf{s},k}\right)_{\mathrm{T}}=v_{\mathbf{s},k}$.
Then define 
\[
w_{k+1}=w_{\mathbf{p},k}A_{i_{k}}b_{i_{k}}u_{i_{k}}w_{\mathbf{s},k}.
\]
Informally, the words $w_{0},\dots,w_{d}$ describe the derivation
of $w_{\mathrm{T}}$ with keeping all the used nonterminals, i.e.,
$A_{i_{k}}$ is rewrited by $A_{i_{k}}b_{i_{k}}u_{i_{k}}$ instead
of $u_{i_{k}}$. Observe that $q_{1}w_{d}\curvearrowright^{*}q_{1}A_{\mathrm{S}}$
using the transitions labeled by words from $P_{\mathrm{BU}}$. Also
observe that, due to Greibach normal form, $w_{d}\in\left(\Sigma_{\mathrm{T}}\cup\Sigma_{\mathrm{T}}\Sigma_{\mathrm{N}}\Sigma_{\mathrm{B}}\right)^{*}$,
which means that the factors from $\Sigma_{\mathrm{N}}\Sigma_{\mathrm{B}}$
are always separated with letters from $\Sigma_{\mathrm{T}}$.

Distinguish the following cases:
\begin{itemize}
\item If $w$ does not have a factor from $\Gamma^{2}\backslash P_{\mathrm{C}}$,
all two-letter factors of $w$ belong to $P_{\mathrm{C}}$, which
implies that $w$ is a factor of some word from $\left(\Sigma_{\mathrm{T}}t\right)^{*}$,
where 
\begin{equation}
t=A_{1}b_{1}A_{2}b_{2}\cdots A_{m}b_{m}.\label{eq:tDef}
\end{equation}

\begin{itemize}
\item If $w$ begins with a letter from $\Sigma_{\mathrm{T}}\cup\Sigma_{\mathrm{N}}$,
and ends with a letter from $\Sigma_{\mathrm{T}}\cup\Sigma_{\mathrm{B}}$,
then $q_{1}w\curvearrowright^{*}q_{1}w_{d}$ using the transitions
labeled by words from $P_{\mathrm{NB}}$. Because $q_{1}w_{d}\curvearrowright^{*}q_{1}A_{\mathrm{S}}$,
we conclude $w\in L\!\left(M_{G}\right)$.
\item Otherwise, $w$ starts with a letter from $\Sigma_{\mathrm{B}}$ or
ends with a letter from $\Sigma_{\mathrm{N}}$. Then 
\[
w_{\mathrm{N,B}}\in\Sigma_{\mathrm{B}}\left(\Sigma_{\mathrm{N}}\Sigma_{\mathrm{B}}\right)^{*}\cup\left(\Sigma_{\mathrm{N}}\Sigma_{\mathrm{B}}\right)^{*}\Sigma_{\mathrm{N}}\cup\Sigma_{\mathrm{B}}\left(\Sigma_{\mathrm{N}}\Sigma_{\mathrm{B}}\right)^{*}\Sigma_{\mathrm{N}}
\]
and we observe that $q_{0}w\curvearrowright q_{3}w\curvearrowright^{*}q_{3}w_{\mathrm{N,b}}\curvearrowright u$
for some $u\in\Sigma_{\mathrm{N}}\cup\Sigma_{\mathrm{b}}\cup\left\{ \varepsilon\right\} $.
As $q_{3}u\curvearrowright q_{4}$, we get $w\in L\!\left(M_{G}\right)$.
\end{itemize}
\item If $w$ has a factor $u\subseteq\Gamma^{2}\backslash P_{\mathrm{C}}$,
write $w=w_{\mathbf{p}}uw_{\mathbf{s}}$ and observe 
\[
w_{\mathbf{p}}q_{0}uw_{\mathbf{s}}\curvearrowright q_{2}w_{\mathbf{p}}w_{\mathbf{s}}\curvearrowright^{*}q_{2},
\]
implying $w\in L\!\left(M_{G}\right)$. 
\end{itemize}

Second, suppose that $L\!\left(M_{G}\right)=\Gamma^{*}$ and take
an arbitrary $v=x_{1}x_{2}\cdots x_{n}\in\Sigma_{\mathrm{T}}^{*}$
with $x_{1},\dots,x_{n}\in\Sigma_{\mathrm{T}}$. Let $w=\left(x_{1}t\right)\left(x_{2}t\right)\cdots\left(x_{n-1}t\right)\left(x_{n}t\right)$,
with $t$ defined in (\ref{eq:tDef}). We have $w\in L\!\left(M_{G}\right)$.
Observe that:
\begin{itemize}
\item The word $w$ does not contain a factor from $\Gamma^{2}\backslash P_{\mathrm{C}}$.
\item By deleting factors from $\Sigma_{\mathrm{B}}\Sigma_{\mathrm{N}}\cup\Sigma_{\mathrm{T}}$,
the word $w$ cannot become a word from $\Sigma_{\mathrm{N}}\cup\Sigma_{\mathrm{B}}\cup\{\epsilon\}$. 
\end{itemize}

Thus, $w$ can be accepted by $M$ only using a path through the state
$q_{1}$ ending in the state $q_{4}$. In other words, $w$ can be
obtained by inserting words from $P_{\mathrm{BU}}\cup P_{\mathrm{NB}}$
to $A_{\mathrm{S}}$. During that process, once an ocurence of some
$b_{i}$ fails to be preceded by $A_{i}$, this situation lasts to
the very end, which is a contradiction. It follows that $b_{i}u_{i}\in P_{\mathrm{BU}}$
can be inserted only to the right of an occurence of $A_{i}$ that
was not followed by $b_{i}$ yet. This corresponds to rewriting $A_{i}$
with $u_{i}$ and we can observe that $w_{\mathrm{T}}=v$ is necessarily
generated by the grammar $G$. 

\end{proof}

\section{Clearing Restarting Automata with Small Contexts\label{sec:Clearing-Restarting-Automata}}

Though the basic model of clearing restarting automata is not able
to describe all context-free languages nor to handle basic language
operations (e.g. concatenation and union) \cite{reaCER2}, it has
been deeply studied in order to design suitable generalizations. The
study considered also restrictions of the maximum context length to
be used in rewriting rules:
\begin{theorem}
[\cite{reaCER2}]~
\begin{enumerate}
\item For each $k\ge3$, the class $\mathcal{L}\!\left(k\cra\right)$ contains
a binary language, which is not context-free. 
\item The class $\mathcal{L}\!\left(2\cra\right)$ contains a language $L\subseteq\Sigma^{\star}$
with $\left|\Sigma\right|=6$, which is not context-free. 
\item The class $\mathcal{L}\!\left(k\cra\right)$ contains only context-free
languages.
\end{enumerate}
\end{theorem}
The present section is devoted to proving the following theorem, which
completes the results listed above.
\begin{theorem}
\label{thm: Main on cl-RA}The class $\mathcal{L}\!\left(2\cra\right)$
contains a binary language, which is not context-free.
\end{theorem}
\noindent In order to prove Theorem \ref{thm: Main on cl-RA}, we
define two particular rewriting systems:
\begin{enumerate}
\item A $1$-context rewriting system $R_{\mathrm{uV}}=\left(\left\{ \mathrm{u},\mathrm{V}\right\} ,\left\{ \mathrm{u},\mathrm{V}\right\} ,I_{\mathrm{uV}}\right)$.
The set $I_{\mathrm{uV}}$ is listed in Table \ref{tab:The-system uV}. 
\item A $2$-clearing restarting automaton $R_{01}=\left(\left\{ 0,1\right\} ,I_{01}\right)$.
The set $I_{\mathrm{uV}}$ is listed in Table \ref{tab:The-system 01}. 
\end{enumerate}
We write $\rightarrow_{\mathrm{uV}}$ for the rewriting relation of
$R_{\mathrm{uV}}$ and $\dashv_{01}$ for the production relation
of $R_{01}$. 

\begin{flushleft}
\begin{table}[h]
\centering{}%
\begin{minipage}[t]{0.3\columnwidth}%
\begin{center}
\begin{tabular}{|c|r|}
\multicolumn{2}{c}{}\tabularnewline
\hline 
0) & $\left(\cent,\epsilon\rightarrow\mathrm{uu},\$\right)$\tabularnewline
\hline 
1) & $\left(\cent,\mathrm{u}\rightarrow\mathrm{uuV},\epsilon\right)$\tabularnewline
\hline 
2) & $\left(\epsilon,\mathrm{Vu}\rightarrow\mathrm{uuuV},\epsilon\right)$\tabularnewline
\hline 
3) & $\left(\epsilon,\mathrm{Vu}\rightarrow\mathrm{uuuu},\$\right)$\tabularnewline
\hline 
\end{tabular}\caption{The rules $I_{\mathrm{uV}}$\label{tab:The-system uV}}

\par\end{center}%
\end{minipage}\hspace*{\fill}%
\begin{minipage}[t]{0.67\columnwidth}%
\begin{center}
\begin{tabular}{|c|r|r|r|r|}
\cline{2-5} 
\multicolumn{1}{c|}{} & \hspace*{\fill}(a)\hspace*{\fill} & \hspace*{\fill}(b)\hspace*{\fill} & \hspace*{\fill}c)\hspace*{\fill} & \hspace*{\fill}d)\hspace*{\fill}\tabularnewline
\hline 
0) & $\left(\cent,00,\$\right)$ & \hspace*{\fill}-\hspace*{\fill} & \hspace*{\fill}-\hspace*{\fill} & \hspace*{\fill}-\hspace*{\fill}\tabularnewline
\hline 
1) & $\left(\cent,10,00\right)$ & $\left(\cent,00,10\right)$ & \hspace*{\fill}-\hspace*{\fill} & \hspace*{\fill}-\hspace*{\fill}\tabularnewline
\hline 
2) & $\left(01,10,00\right)$ & $\left(00,11,01\right)$ & $\left(11,00,10\right)$ & $\left(10,01,11\right)$\tabularnewline
\hline 
3) & $\left(01,10,0\$\right)$ & $\left(00,11,0\$\right)$ & \hspace*{\fill}-\hspace*{\fill} & \hspace*{\fill}-\hspace*{\fill}\tabularnewline
\hline 
\end{tabular}\caption{The rules $I_{01}$ sorted by types 0 to 3\label{tab:The-system 01}}

\par\end{center}%
\end{minipage}
\end{table}

\par\end{flushleft}

The key feature of the system $R_{\mathrm{uV}}$ is:
\begin{lemma}
\label{lem: only powers in uV cup u-star}Let $w\in L\!\left(R_{\mathrm{uV}}\right)\cap\left\{ \mathrm{u}\right\} ^{\star}$.
Then $\left|w\right|=2\cdot3^{n}$ for some $n\ge0$.
\end{lemma}
The proof is postponed to Section \ref{sub:Proof-of-Theorem shift double defect}.
We also define:
\begin{enumerate}
\item A length-preserving mapping $\varphi:\left\{ 0,1\right\} ^{\star}\rightarrow\left\{ \mathrm{u},\mathrm{V}\right\} ^{\star}$
as $\varphi\!\left(x_{1}\dots x_{n}\right)=\overline{x}_{1}\dots\overline{x}_{n}$,
where
\[
\overline{x}_{k}=\begin{cases}
\mathrm{V} & \mbox{if }1<k<n\mbox{ and }x_{k-1}=x_{k+1}\\
\mathrm{u} & \mbox{otherwise}
\end{cases}
\]
for each $k\in\left\{ 1,\dots,n\right\} $.
\item A regular language $K\subseteq\left\{ 0,1\right\} ^{\star}$:
\begin{eqnarray*}
K & = & \left\{ w\in\left\{ 0,1\right\} ^{\star}\mid w\mbox{ has none of the factors }000,010,101,111\right\} .
\end{eqnarray*}

\end{enumerate}
The following is a trivial property of $\varphi$ and $K$:
\begin{lemma}
\label{lem: defect free eq to u-only}Let $u\in\left\{ 0,1\right\} ^{\star}$.
Then $u\in K$ if and only if $\varphi\!\left(u\right)\in\left\{ \mathrm{u}\right\} ^{\star}$.
\end{lemma}
The next lemma expresses how the systems $R_{01}$ and $R_{\mathrm{uV}}$
are related:

\begin{lemma}
\label{lem: if 01 then uV}Let $u,v\in\left\{ 0,1\right\} ^{\star}$.
If $u\dashv_{01}v$, then $\varphi\!\left(u\right)\rightarrow_{\mathrm{uV}}\varphi\!\left(v\right)$.\end{lemma}
\begin{proof}
For $u=v$ the claim is trivial, so we suppose $u\neq v$. Denote
$m=\left|u\right|$. As $u$ can be rewrote to $v$ using a single
rule of $R_{01}$, we can distinguish which of the four kinds of rules
(the rows 0 to 3 of Table \ref{tab:The-system 01}) is used:
\begin{enumerate}
\item [0)]If the rule 0 is used, we have $u=\epsilon$ and $v=00$. Thus
$\varphi\!\left(u\right)=\epsilon$ and $\varphi\!\left(v\right)=\mathrm{uu}$.
\item [1)]If a rule $\left(\cent,z_{1}z_{2},y_{1}y_{2}\right)$ of the
kind 1 is used, we see that $v$ has some of the prefixes $1000,0010$
and so $\varphi\!\left(v\right)$ starts with $\mathrm{uuV}$. Trivially,
$\varphi\!\left(u\right)$ starts with $\mathrm{u}$. Because $u\!\left[1..\right]=v\!\left[3..\right]$,
we have $\varphi\!\left(u\right)\!\left[2..\right]=\varphi\!\left(v\right)\!\left[4..\right]$
and we conclude that applying the rule $\left(\cent,\mathrm{u}\rightarrow\mathrm{uuV},\epsilon\right)$
rewrites $\varphi\!\left(u\right)$ to $\varphi\!\left(v\right)$. 
\item [2)]If a rule $\left(x_{1}x_{2},z_{1}z_{2},y_{1}y_{2}\right)$ of
the kind 2 is used, we have
\begin{eqnarray*}
u\!\left[k..k+3\right] & = & x_{1}x_{2}y_{1}y_{2},\\
v\!\left[k..k+5\right] & = & x_{1}x_{2}z_{1}z_{2}y_{1}y_{2}.
\end{eqnarray*}
for some $k\in\left\{ 1,\dots,m-3\right\} $. As $x_{1}x_{2}y_{1}y_{2}$
equals some of the factors $0100,0001,1110,1011$, we have
\[
\varphi\!\left(u\right)\!\left[k+1..k+2\right]=\mathrm{Vu}.
\]
As $x_{1}x_{2}z_{1}z_{2}y_{1}y_{2}$ equals some of the factors $011000,001101,110010,100111$,
we have
\[
\varphi\!\left(v\right)\!\left[k+1..k+4\right]=\mathrm{uuuV}.
\]
Because $u\!\left[..k+1\right]=v\!\left[..k+1\right]$ and $u\!\left[k+2..\right]=v\!\left[k+4..\right]$,
we have
\begin{eqnarray*}
\varphi\!\left(u\right)\!\left[..k\right] & = & \varphi\!\left(v\right)\!\left[..k\right],\\
\varphi\!\left(u\right)\!\left[k+3..\right] & = & \varphi\!\left(v\right)\!\left[k+5..\right].
\end{eqnarray*}
Now it is clear that the rule $\left(\epsilon,\mathrm{Vu}\rightarrow\mathrm{uuuV},\epsilon\right)$
rewrites $\varphi\!\left(u\right)$ to $\varphi\!\left(v\right)$.
\item [3)]If a rule $\left(x_{1}x_{2},z_{1}z_{2},y\$\right)$ of the kind
3 is used, we have
\begin{eqnarray*}
u\!\left[m-2..m\right] & = & x_{1}x_{2}y,\\
v\!\left[m-2..m+2\right] & = & x_{1}x_{2}z_{1}z_{2}y.
\end{eqnarray*}
As $x_{1}x_{2}y$ equals some of the factors $010,000$, we have
\begin{eqnarray*}
\varphi\!\left(u\right)\!\left[m-1..m\right] & = & \mathrm{Vu}.
\end{eqnarray*}
As $x_{1}x_{2}z_{1}z_{2}y$ equals some of the factors $01100,00110$,
we have
\begin{eqnarray*}
\varphi\!\left(v\right)\!\left[m-1..m+2\right] & = & \mathrm{uuuV}.
\end{eqnarray*}
Because $u\!\left[..m-1\right]=v\!\left[..m-1\right]$, we have
\begin{eqnarray*}
\varphi\!\left(u\right)\!\left[..m-2\right] & = & \varphi\!\left(v\right)\!\left[..m-2\right],
\end{eqnarray*}
Now it is clear that the rule $\left(\epsilon,\mathrm{Vu}\rightarrow\mathrm{uuuu},\$\right)$
rewrites $\varphi\!\left(u\right)$ to $\varphi\!\left(v\right)$.
\end{enumerate}
\end{proof}
\begin{corollary}
If $u\in L\!\left(R_{01}\right)$, then $\varphi\!\left(u\right)\in L\!\left(R_{\mathrm{uV}}\right)$.\end{corollary}
\begin{proof}
Follows from the fact that $\varphi\!\left(\epsilon\right)=\epsilon$
and a trivial inductive use of Lemma \ref{lem: if 01 then uV}.
\end{proof}
The last part of the proof of Theorem relies of the following lemma,
whose proof is postponed to Section \ref{sub:Proof-of-Theorem shift double defect}:
\begin{lemma}
\label{lem: shift double defect}For each $\alpha,\beta>0$ it holds
that
\[
00\left(1100\right)^{\alpha}1000\left(1100\right)^{\beta}\dashv_{01}00\left(1100\right)^{\alpha+9}1000\left(1100\right)^{\beta-1}.
\]
\end{lemma}
\begin{corollary}
\label{cor:full shift double defect}For each $\beta>0$ it holds
that
\[
001000\left(1100\right)^{\beta}\dashv_{01}00\left(1100\right)^{9\beta}1000.
\]
\end{corollary}
\begin{proof}
As the left-hand side is equal to $00\left(1100\right)^{0}1000\left(1100\right)^{\beta}$
and the right-hand side is equal to $00\left(1100\right)^{9\beta}1000\left(1100\right)^{0}$,
the claim follows from an easy inductive use of Lemma \ref{lem: shift double defect}.
\end{proof}

\begin{corollary}
\label{cor:L cup K is infinite}The language $L\!\left(R_{01}\right)\cap K$
is infinite.\end{corollary}
\begin{proof}
We show that for each $k\ge0$, 
\begin{eqnarray*}
00\left(1100\right)^{\frac{2\cdot9^{k}-2}{4}} & \in & L\!\left(R_{01}\right).
\end{eqnarray*}
In the case $k=0$ we just check that $00\in L\!\left(R_{01}\right)$.
Next we suppose that the claim holds for a fixed $k\ge0$ and show
that
\[
00\left(1100\right)^{\frac{2\cdot9^{k}-2}{4}}\dashv_{01}00\left(1100\right)^{\frac{2\cdot9^{k+1}-2}{4}}.
\]
Using the rules 1a and 1b we get
\[
00\left(1100\right)^{\frac{2\cdot9^{k}-2}{4}}\dashv_{01}1000\left(1100\right)^{\frac{2\cdot9^{k}-2}{4}}\dashv_{01}001000\left(1100\right)^{\frac{2\cdot9^{k}-2}{4}},
\]
while Corollary \ref{cor:full shift double defect} continues with
\[
001000\left(1100\right)^{\frac{2\cdot9^{k}-2}{4}}\dashv_{01}00\left(1100\right)^{\frac{2\cdot9^{k+1}-18}{4}}1000.
\]
Finally, denoting $p=00\left(1100\right)^{\frac{2\cdot9^{k+1}-18}{4}}$,
using rules 2b, 2a, 2b, 2d, 2c, and 2a respectively we get
\[
\begin{aligned}p1000\dashv_{01}p100\underline{11}0\dashv_{01}p1\underline{10}00110\dashv_{01}p1100\underline{11}0110\dashv_{01}p1100110\underline{01}110\dashv_{01}\\
\dashv_{01}p1100110011\underline{00}10\dashv_{01}p1100110011001\underline{10}0=00\left(1100\right)^{\frac{2\cdot9^{k+1}-2}{4}}.
\end{aligned}
\]

\end{proof}
We conclude the proof of Theorem \ref{thm: Main on cl-RA} by pointing
out that Lemmas \ref{lem: defect free eq to u-only}, \ref{lem: if 01 then uV},
and \ref{lem: only powers in uV cup u-star} say that for each $w\in\left\{ 0,1\right\} ^{\star}$
we have 
\begin{eqnarray*}
w\in L\!\left(R_{01}\right)\cap K & \Rightarrow & \varphi\!\left(w\right)\in L\!\left(R_{\mathrm{uV}}\right)\cap\left\{ \mathrm{u}\right\} ^{\star}\\
 & \Rightarrow & \left(\exists n\ge0\right)\,\left|w\right|=2\cdot3^{n}
\end{eqnarray*}
This, together with the pumping lemma for context-free languages and
the infiniteness of $L\!\left(R_{01}\right)\cap K$, implies that
$L\!\left(R_{01}\right)\cap K$ is not a context-free language. As
the class of context-free languages is closed under intersections
with regular languages, nor $L\!\left(R_{01}\right)$ is context-free.

\subsection{\label{sub:Proof-of-Theorem shift double defect}Proofs of Lemmas
\ref{lem: only powers in uV cup u-star} and \ref{lem: shift double defect}}
\begin{proof}
[of Lemma \ref{lem: only powers in uV cup u-star}]We should prove
that $w\in L\!\left(R_{\mathrm{uV}}\right)\cap\left\{ \mathrm{u}\right\} ^{\star}$
implies $\left|w\right|=2\cdot3^{n}$ for some $n\ge0$. Let $\Phi:\left\{ \mathrm{u},\mathrm{V}\right\} ^{\star}\rightarrow\mathbb{N}$
be defined inductively as follows:
\begin{eqnarray*}
\Phi\!\left(\epsilon\right) & = & 0,\\
\Phi\!\left(\mathrm{u}^{k}w\right) & = & k+\Phi\!\left(w\right),\\
\Phi\!\left(\mathrm{V}w\right) & = & 1+3\cdot\Phi\!\left(w\right)
\end{eqnarray*}
for each $k\ge1$ and $w\in\left\{ \mathrm{u},\mathrm{V}\right\} ^{\star}$.
Observe that we have assigned a unique value of $\Phi$ to each word
from $\left\{ \mathrm{u},\mathrm{V}\right\} ^{\star}$. Next, we describe
effects of the rules of $R_{\mathrm{uV}}$ to the value of $\Phi$.
\begin{enumerate}
\item [0)]The rule $0$ can only rewrite $w_{1}=\epsilon$ to $w_{2}=\mathrm{uu}$.
We have $\Phi\!\left(w_{1}\right)=0$ and $\Phi\!\left(w_{2}\right)=2$.
\item [1)]The rule $1$ rewrites $w_{1}=\mathrm{u}w$ to $w_{2}=\mathrm{uuV}w$
for some $w\in\left\{ \mathrm{u},\mathrm{V}\right\} ^{\star}$. We
have $\Phi\!\left(w_{1}\right)=1+\Phi\!\left(w\right)$ and $\Phi\!\left(w_{2}\right)=3+3\cdot\Phi\!\left(w\right)$.
Thus, $\Phi\!\left(w_{2}\right)=3\cdot\Phi\!\left(w_{1}\right)$.
\item [2)]The rule $2$ rewrites $w_{1}=\overline{w}\mathrm{Vu}w$ to $w_{2}=\overline{w}\mathrm{uuuV}w$
for some $w,\overline{w}\in\left\{ \mathrm{u},\mathrm{V}\right\} ^{\star}$.
We have 
\[
\Phi\!\left(\mathrm{Vu}w\right)=\Phi\!\left(\mathrm{uuuV}w\right)=4+3\cdot\Phi\!\left(w\right).
\]
It follows that $\Phi\!\left(w_{1}\right)=\Phi\!\left(w_{2}\right)$.
\item [3)]The rule $3$ rewrites $w_{1}=\overline{w}\mathrm{Vu}$ to $w_{2}=\overline{w}\mathrm{uuuu}$
for some $\overline{w}\in\left\{ \mathrm{u},\mathrm{V}\right\} ^{\star}$.
We have $\Phi\!\left(\mathrm{Vu}\right)=\Phi\!\left(\mathrm{uuuu}\right)=4$
and thus $\Phi\!\left(w_{1}\right)=\Phi\!\left(w_{2}\right)$.
\end{enumerate}

Together, each $w\in L\!\left(R_{\mathrm{uV}}\right)$ has $\Phi\!\left(w\right)=2\cdot3^{n}$
for some $n\ge0$. As $\Phi\!\left(w\right)=\left|w\right|$ for each
$w\in\left\{ \mathrm{u}\right\} ^{\star}$, the proof is complete.

\end{proof}

\begin{proof}
[of Lemma \ref{lem: shift double defect}]We should prove that
\[
00\left(1100\right)^{\alpha}1000\left(1100\right)^{\beta}\dashv_{01}00\left(1100\right)^{\alpha+9}1000\left(1100\right)^{\beta-1}.
\]
for each $\alpha,\beta>0$. Indeed:
\[
\begin{aligned}00\left(1100\right)^{\alpha}1000\left(1100\right)^{\beta} & \,\,\dashv_{01}\,\,00\left(1100\right)^{\alpha}100\underline{11}0\left(1100\right)^{\beta} & \dashv_{01}\\
00\left(1100\right)^{\alpha}1\underline{10}00110\left(1100\right)^{\beta} & \,\,\dashv_{01}\,\,00\left(1100\right)^{\alpha+1}\underline{11}0110\left(1100\right)^{\beta} & \dashv_{01}\\
00\left(1100\right)^{\alpha+1}110\underline{01}110\left(1100\right)^{\beta} & \,\,\dashv_{01}\,\,00\left(1100\right)^{\alpha+2}1110\underline{01}\left(1100\right)^{\beta} & \dashv_{01}\\
00\left(1100\right)^{\alpha+2}11\underline{00}1001\left(1100\right)^{\beta} & \,\,\dashv_{01}\,\,00\left(1100\right)^{\alpha+3}1\underline{10}001\left(1100\right)^{\beta} & \dashv_{01}\\
00\left(1100\right)^{\alpha+4}\underline{11}01\left(1100\right)^{\beta} & \,\,\dashv_{01}\,\,00\left(1100\right)^{\alpha+4}11011\underline{00}100\left(1100\right)^{\beta-1} & \dashv_{01}\\
00\left(1100\right)^{\alpha+4}110\underline{01}1100100\left(1100\right)^{\beta-1} & \,\,\dashv_{01}\,\,00\left(1100\right)^{\alpha+5}11\underline{00}100100\left(1100\right)^{\beta-1} & \dashv_{01}\\
00\left(1100\right)^{\alpha+6}1\underline{10}00100\left(1100\right)^{\beta-1} & \,\,\dashv_{01}\,\,00\left(1100\right)^{\alpha+7}01\underline{10}00\left(1100\right)^{\beta-1} & \dashv_{01}\\
00\left(1100\right)^{\alpha+7}\underline{11}011000\left(1100\right)^{\beta-1} & \,\,\dashv_{01}\,\,00\left(1100\right)^{\alpha+7}110\underline{01}11000\left(1100\right)^{\beta-1} & \dashv_{01}\\
00\left(1100\right)^{\alpha+8}11\underline{00}1000\left(1100\right)^{\beta-1} & .
\end{aligned}
\]

\end{proof}
\bibliographystyle{C:/Users/Vojta/Desktop/SYNCHRO2/bib/splncs03}
\bibliography{C:/Users/Vojta/Desktop/SYNCHRO2/bib/ruco}

\end{document}